\documentclass[12pt]{article}

\usepackage{amsmath}
\usepackage{amssymb}
\usepackage{amsthm}
\usepackage{mathrsfs}

\usepackage{slashed}
\usepackage{color}
\usepackage{esint}

\usepackage{graphicx}
\usepackage[all,cmtip]{xy}
\usepackage{hyperref}

\newtheorem{thm}{Theorem}
\newtheorem{lemma}{Lemma}
\newtheorem{prop}{Proposition}

\newtheorem{defn}{Definition}
\newtheorem{examp}{Example}

\hypersetup{
    pdfborder = {0 0 0},%
    colorlinks,%
    citecolor=blue,%
    filecolor=black,%
    linkcolor=red,
    urlcolor=green
}

\begin{document}

\title{The Stiefel--Whitney theory of \\
         topological insulators}
\author{ Ralph M. Kaufmann \thanks{e-mail: rkaufman@math.purdue.edu  }
           \and
         Dan Li  \thanks{ e-mail: li1863@math.purdue.edu } \\
            \and
        Birgit  Wehefritz-Kaufmann   \thanks{  e-mail: ebkaufma@math.purdue.edu     }
         Stiefel--Whitney class, cobordism theory, extended TQFT   
	    }

\date{}
\maketitle

\begin{abstract} 
We study the topological band theory of time reversal invariant topological insulators
and interpret the topological $\mathbb{Z}_2$ invariant as an obstruction in terms of Stiefel--Whitney classes.  
The band structure of a topological insulator defines a Pfaffian line bundle over the momentum space, whose
 structure group can be reduced to $\mathbb{Z}_2$.
So the  topological $\mathbb{Z}_2$ invariant will be understood by the Stiefel--Whitney theory, 
which detects the orientability of a principal $\mathbb{Z}_2$-bundle.
Moreover, the relation between weak and strong topological insulators will be understood  
based on cobordism theory. Finally, 
the topological $\mathbb{Z}_2$ invariant gives rise to a fully extended topological quantum field theory (TQFT).
\end{abstract}

\section{Introduction}\label{Intro}
Time reversal invariant topological insulators, or simply topological insulators,
are new materials with conducting edge states protected by the time reversal $\mathbb{Z}_2$ symmetry and charge conservation (a $U(1)$ symmetry) \cite{HK10}.
The electronic band structure of a topological insulator
defines  a vector bundle  over the 
momentum space, which is called the Bloch bundle by physicists \cite{KLW15}.
The Bloch bundle is naturally equipped with a Hermitian metric, 
so it becomes an even rank Hilbert bundle, i.e., each fiber is an inner product vector space.
The characteristic feature of a topological insulator  is determined by its top band, i.e., the top subbundle of the Bloch bundle.
So we are mostly interested in this top subbundle and call it the Hilbert bundle,  which is a rank 2 complex vector bundle. 
This Hilbert bundle is usually the main geometric object to study in the band theory of a topological insulator.

Because of the time reversal $\mathbb{Z}_2$-symmetry, topological insulators are characterized by a $\mathbb{Z}_2$-valued invariant, 
which is  called the topological $\mathbb{Z}_2$ invariant.
There are many equivalent characterizations  of the topological $\mathbb{Z}_2$ invariant, 
we will use the Kane--Mele invariant as its definition in this paper \cite{FK06}.
The physical meaning of the topological $\mathbb{Z}_2$ invariant is the existence of an unpaired Majorana zero mode, 
i.e., geometrically a conical singularity created by the intersection of edge states. In the Hilbert bundle, 
intersections of edge states correspond to intersections of sections with opposite orientation.

Based on the topological band theory, we will focus on the concept of orientability and study the topological $\mathbb{Z}_2$ invariant 
 from the perspective of obstruction theory.
One main result of this paper is that the topological $\mathbb{Z}_2$ invariant can be understood by Stiefel--Whitney classes of a Pfaffian line bundle over the momentum space.
To this end, we will introduce the associated Pfaffian line bundle of a topological insulator, which will be viewed as a \emph{real} line bundle
 and compared to the  M{\"o}bius strip.  So instead of the Hilbert bundle, we will focus on the Pfaffian line bundle.
Physically, a chiral zero mode 
(globally represented by the Pfaffian line bundle) together with the time reversal symmetry completely determine a Majorana zero mode
(globally represented by the Hilbert bundle). 
As a principal $\mathbb{Z}_2$-bundle, the Pfaffian line bundle is topologically characterized by its Stiefel--Whitney classes.
Therefore, the topological $\mathbb{Z}_2$ invariant is nonzero if and only if the Stiefel--Whitney class of the Pfaffian line bundle is nonzero. 

Furthermore, we will try to understand the topological $\mathbb{Z}_2$ invariant
in the context of cobordism. First of all, the momentum space $X$ is viewed as a CW complex with a $\mathbb{Z}_2$ 
action defined by the time reversal symmetry, i.e., a $\mathbb{Z}_2$-CW complex.
Then the $\mathbb{Z}_2$-cell decomposition of $X$ induces cobordant relations between $k$ and $(k+1)$-cells of $X$, and the Pfaffian line bundle over $X$ 
is restricted to  $k$ or $(k+1)$-cells. When $X $ is a $n$-torus, it can be viewed as a cobordism with the boundary given by two sub-tori.
In this case, one important observation is that
the topological $\mathbb{Z}_2$ invariant is trivial if and only if  
the total spaces of the restricted Pfaffian bundles over two sub-tori are cobordant. And this is equivalent to the fact that 
the Stiefel--Whitney classes of the restricted Pfaffian bundles are the same.
As a  consequence,  the relation between the weak  and strong topological insulators is understood as a cobordant relation.
Moreover, a cobordism also induces an addition between the Stiefel--Whitney classes of its boundaries, 
which gives rise to a higher order Stiefel--Whitney class of the cobordism itself.

In other words, the momentum space $X$ can be viewed as a bordism category based on the $\mathbb{Z}_2$-CW decomposition.
So the topological $\mathbb{Z}_2$ invariant defines a fully extended topological quantum field theory (TQFT), 
i.e., a symmetric monoidal functor from the bordism category to the category of vector spaces.
For $3$-dimensional topological insulators, this fully extended TQFT is completely determined by the $0$-dimensional fixed points of the time reversal symmetry.
   
This paper is organized as follows.  In section \ref{TopBand}, we review the topological band theory of a topological insulator and 
interpret the topological $\mathbb{Z}_2$ invariant as an obstruction to orientability using Stiefel--Whitney classes.
In section \ref{Cobthy}, we use the knowledge of unoriented cobordism to understand the topological $\mathbb{Z}_2$ invariant and the  
relation between weak  and strong topological insulators.   In section \ref{extTQFT}, we reformulate the 
topological $\mathbb{Z}_2$ invariant as a fully extended topological quantum field theory.
For convenience, we put some background knowledge 
such as the M{\"o}bius strip, Stiefel--Whitney class, Berry phase and Kane--Mele invariant in the appendices \ref{Appdx}.

 \section{Topological band theory}\label{TopBand}
 In this section, we will review the topological band theory of a topological insulator, that is, 
 the relevant vector bundles defined by the electronic band structure in the presence of the time reversal symmetry. 
 And the emphasis will be placed on the associated Pfaffian line bundle, whose structure group will be reduced to $\mathbb{Z}_2$, i.e., $\mathbb{Z}/2\mathbb{Z}$.
 Next we will analyze the $1$-dimensional case carefully, and compare the Pfaffian line bundle over the circle to the M{\"o}bius strip.
 As a result, the topological $\mathbb{Z}_2$ invariant 
 will be understood as an obstruction to the orientability of the Pfaffian line bundle using Stiefel--Whitney classes in different dimensions.
 
 \subsection{Pfaffian line bundle}
 In condensed matter physics, one studies the lattice model of a material,
 and investigates the effective Hamiltonian of some (quasi-)particles such as electrons defined over the lattice.
 If one only considers the translational symmetry of a lattice, 
 then the Pontryagin dual of this lattice 
 is the momentum space. Let $X$ be a compact space  without boundary representing the momentum space. 
 For example, the torus $\mathbb{T}^d$ is the momentum space of the lattice $\mathbb{Z}^d$, and the sphere $\mathbb{S}^d$ 
 is viewed as the momentum space  of the continuous limit of $\mathbb{Z}^d$, i.e., $\mathbb{R}^d$. The torus $\mathbb{T}^d$ and 
 the sphere $\mathbb{S}^d$ will be the main examples in this paper.

Time reversal symmetry is a $\mathbb{Z}_2$ symmetry that basically changes the direction of time. On a momentum space $X$, 
time reversal symmetry defines
an involutive time reversal transformation $\tau$, that is, a homeomorphism flips the sign of a local coordinate,
$$
\tau: X \rightarrow X; \quad x \mapsto -x, \quad \text{with} \,\,\,\, \tau^2 = id_X
$$
So the momentum space $X$ becomes an involutive space $(X, \tau)$ with the involution $\tau$, 
which is also called a Real space with the real structure $\tau$.

 The set of fixed points of the time reversal transformation is denoted by,
 $$
 X^\tau = \{ x \in X ~|~ \tau(x) = x\}
 $$
 Since $X$ is a compact space without boundary, we assume that $X^\tau$ is a finite set with even number of isolated fixed points. 
 $X^\tau$ is also viewed as the set of \emph{real} points with respect to the real structure $\tau$.
 
 The involutive space $(X, \tau)$ has the structure of a $\mathbb{Z}_2$-CW complex, 
 that is, there is a $\mathbb{Z}_2$-equivariant cellular decomposition of $X$, see \cite{DG14} for concrete examples. In the other way around,
 starting with the fixed points $X^\tau$, $X$ can be built up by gluing cells that carry a free $\mathbb{Z}_2$ action, i.e., $\mathbb{Z}_2$-cells. 
 Such equivariant cellular decomposition is very useful in the computation of cohomology or K-theory of $X$ \cite{KLW1501}.
 This construction is closely related to the stable splitting of $X$ into spheres respecting the time reversal symmetry \cite{FM13}.
  
 Let  $H$ be a time reversal invariant Hamiltonian of a topological insulator, 
 and $\psi$ be an eigenstate representing an electronic state.
Time reversal symmetry defines 
an anti-unitary operator $\Theta$, called the time reversal operator, acting on  $\psi$ with the property $\Theta^2 = -1$,
 so that the states $\psi$ and $\Theta \psi$ are orthogonal to each other, i.e., $\langle \psi, \Theta \psi \rangle = 0$.
 Either $\psi$ or $\Theta \psi$ is called a chiral state, and they have opposite chirality 
 since time reversal symmetry reverses the direction of a state, that is, if 
 $\psi$ is  right-moving, then $\Theta \psi$ is  left-moving.
 With the time reversal operator $\Theta$ and the time reversal transformation $\tau$, a time reversal invariant Hamiltonian $H$
 must satisfy the condition $\Theta H(x) \Theta^* = H(\tau(x))$.
 
 In the eigenvalue equation of the Hamiltonian $H$, i.e., $H \psi_n = E_n \psi_n$, the energy level $E_n : X \rightarrow \mathbb{R}$ is 
 called the $n$-th band function. The band structure of a topological insulator 
 is defined by the band functions $\{ E_n,  1 \leq n \leq N \} $ for all filled bands,
 where $N$ is a positive integer labels the top occupied band. 
 For simplicity, we assume that the band structure is non-degenerate, that is, there is no 
 band crossing between different bands. However, the chiral states $\psi_n$ and $\Theta \psi_n$ have the same energy level $E_n$
 so that each band is doubly degenerate for a time reversal invariant topological insulator.
 
 \begin{defn}
   The band structure of a topological insulator defines a vector bundle over the momentum space with even rank, denoted by $\pi: \mathcal{B} \rightarrow X$,
   which is called the Bloch bundle by physicists. 
   
 \end{defn}
 
 Physically, the Bloch bundle models the finite dimensional physical Hilbert space of electronic states in the occupied bands. 
 So the Bloch bundle is naturally equipped with a Hermitian metric derived from the physical Hilbert space. 
 In other words, $\mathcal{B}$ is a Hilbert bundle of finite rank, 
 i.e., each fiber is an inner product vector space.
 By the non-degeneracy of the Bloch bundle, it can be decomposed as a Whitney sum of subbundles, i.e., 
 $\mathcal{B} = \oplus_{i = 1}^N \mathcal{B}_i $ where $N$ is the number of filled bands.
 Since the characteristic feature of a topological insulator is completely determined by its top occupied band, 
 we only care about the top subbundle $\mathcal{B}_N $.
 
 \begin{defn}
  The Hilbert bundle $\pi: \mathcal{H} \rightarrow X$ is defined as the top Bloch subbundle, i.e., $\mathcal{H} = \mathcal{B}_N$.
  \end{defn}

 So the rank of the Hilbert bundle $ \mathcal{H}$ is 2, and 
  physical states in the top band are modeled by  elements in the space of sections $\Gamma(X, \mathcal{H})$ up to a phase.
 In a small neighborhood of a fixed point, 
 a typical local section of $\mathcal{H}$ represents a pair of physical states $(\psi, \Theta \psi)$ 
 consisting of an electronic state $\psi$ and its mirror image $\Theta \psi$ under the time reversal symmetry.

 Time reversal symmetry defines a pullback on the space of sections induced by $\Theta$ and $\tau$,  
 $$
  \sigma : \Gamma(X, \mathcal{H}) \rightarrow \Gamma(X, \mathcal{H}); \quad s \mapsto - \Theta \circ s \circ \tau \quad \text{with} \quad \sigma^2 = -1
 $$
 So one can compare two sections $s$ and $\sigma s$ since they have the same domain. 
 We have to point out that in general $s$ and $\Theta s$ have different domains.
 
 \begin{defn}
    The transition function of the Hilbert bundle $\mathcal{H}$ is defined by 
 \begin{equation}\label{TransFunc}
   w : X \rightarrow U(2); \quad w(x) := \langle s(x), \sigma t(x)\rangle 
 \end{equation}
 where $s, t \in \Gamma(X, \mathcal{H})$ are local sections. 
 \end{defn} 
 In fact, the transition function is first defined over the overlaps of coordinate patches, then extended to the whole momentum space.
 By the $\mathbb{Z}_2$-CW complex structure of the involutive space $(X, \tau)$, the transition function $w$ is entirely determined by
 the behavior around the fixed points $X^\tau$. Therefore, we can assume that the Hilbert bundle $\pi: \mathcal{H} \rightarrow X$ is a 
 locally trivial bundle with possible degeneracy only at $X^\tau$. In other words, $\mathcal{H}$ 
 is obtained from the trivial bundle $X \times \mathbb{C}^2$ by twisting the fibers around some fixed points, 
 and such twists are kept track of by the transition function $w$. Under this assumption, $\psi$ and $\Theta \psi$ can be viewed as
 global states with possible intersections only at some fixed points.
 
 The transition function of the Hilbert bundle has an important property derived from 
 $\Theta^2 = -1$ \cite{KLW15}, 
 $$
 w(x)^T = -w(\tau(x)), \quad  x \in X
 $$ where $T$ stands for the transpose of a matrix. 
 In particular, it becomes skew-symmetric at a fixed point, 
 $$
 w(x)^T = -w(x), \quad  x \in X^\tau
 $$
 Recall that if one takes the Pfaffian function of a $2 \times 2$ skew-symmetric matrix,
 then one picks up the right-upper corner entry of the matrix. The Pfaffian function suggests a similar operation in the geometric picture.

 \begin{defn}\label{Pfdef}
  The Pfaffian line bundle $\pi: (Pf, \Theta) \rightarrow X$, or simply $ \pi: Pf \rightarrow X$, is defined by half of the Hilbert bundle, 
  that is, one takes a chiral state   $\psi$ or $\Theta \psi$ from the pair $(\psi, \Theta \psi)$. The transition function of the Pfaffian line bundle is
  given by the Pfaffian of $w$, i.e., 
  \begin{equation}
      h : = pf(w) : X \rightarrow U(1)
  \end{equation}
  where $h$ is first defined over $X^\tau$, then extended to $X$ trivially.
 \end{defn}
 So the Pfaffian line bundle $\pi: Pf \rightarrow X$ can be obtained from the trivial bundle $X \times \mathbb{C}$ by multiplying a phase at the fixed points.
 If one combines the Pfaffian line bundle   with the time reversal symmetry, 
 one easily recovers  the Hilbert bundle. Roughly speaking, a chiral state together with the time reversal symmetry determine the pair of chiral states, and vice versa,
 $$
  (\psi, \Theta ) \Longleftrightarrow (\psi, \Theta \psi) \quad \text{or} \quad (\Theta \psi, \Theta ) \Longleftrightarrow (\psi, \Theta \psi)
 $$
 In other words, keeping time reversal symmetry in mind, 
 the Pfaffian line bundle is a reduction of the Hilbert bundle. 
 We will see later that the structure group of the Pfaffian line bundle will be further reduced from $U(1)$ to $\mathbb{Z}_2$,
 so that $\pi: Pf \rightarrow X$ can be compared to the M{\"o}bius strip. 
 If one wants to emphasize the $\mathbb{Z}_2$ structure group and its relation to orientability,
 then one can employ the canonical principal $\mathbb{Z}_2$-bundle, i.e., the associated orientation bundle of the Pfaffian line bundle.
 
 This definition of the Pfaffian line bundle is slightly different from the classical construction of a Pfaffian line bundle. 
 The main difference is that according to the classical construction a Pfaffian line bundle is a real line bundle,
 but in our case $\pi: Pf \rightarrow X$ is a complex line bundle.
 However, the time reversal operator $\Theta$ is a general real structure such that $\Theta^2 = -1$, 
 so $\pi: Pf \rightarrow X$ is a \emph{real} line bundle with respect to $\Theta$. 
 Thus, the key to  understanding 
 the topological band theory of a topological insulator is that the Pfaffian line bundle $\pi: Pf \rightarrow X$ is viewed as a \emph{real} line bundle 
 with the structure group $\mathbb{Z}_2$.

 \subsection{One dimensional case}
 In this subsection, we will study the Hilbert bundle and the Pfaffian line bundle over the circle,
   and interpret the topological $\mathbb{Z}_2$ invariant as an obstruction to orientability by comparing the Pfaffian line bundle to the M{\"o}bius strip. 
 
 For the one dimensional momentum space $X = \mathbb{T} = \mathbb{S}^1$, one has the skeleton decomposition of $(\mathbb{T}, \tau)$ as, 
 \begin{equation} \label{1dCW}
     \mathbb{T} = \mathbb{T}^\tau \sqcup (D^1 \times \mathbb{Z}_2) = \{ 0, \pi \} \sqcup ((0, \pi) \cup (-\pi , 0))
 \end{equation}
 provided that $\mathbb{T}$ is parametrized by the angle.
 So the fixed points $\mathbb{T}^\tau = \{ 0, \pi \}$ divide the circle into two open regions: $R_I = (0, \pi)$ and $R_{II} = \tau(0, \pi) = (-\pi, 0)$.
 Notice that the closure of the first region $\overline{R}_I = [0, \pi]$ is also called the effective Brillouin zone, and  the whole 
 Brillouin zone is $\mathbb{T}$ in 1d. 
 
 Let us work out the transition function of the Hilbert bundle $\pi: \mathcal{H} \rightarrow \mathbb{T}$. 
 Pick a local section $ (\phi, \chi) \in \Gamma(\mathbb{T}, \mathcal{H})$ in a small neighborhood $O_0$ of the fixed point $0$,
 we restrict them to the open region $R_I$ (resp. $R_{II}$) and denote them by $(\phi_I, \chi_I)$ (resp.  $(\phi_{II}, \chi_{II})$). 
 By time reversal symmetry, $\chi_I$ and the pullback $\sigma \phi_{II}$ are the same up to a phase,
 so one has the relation, 
 \begin{equation} 
    \chi_I(x) = e^{i \beta(x)}\sigma \phi_{II}(x), \quad x \in R_I
 \end{equation}
 where $\beta: \mathbb{T} \rightarrow \mathbb{R}$ is assumed to be a continuous phase function. 
 Physically, $\beta$ is a gauge transformation that satisfies the constraint \eqref{betaCond} in the appendix \ref{BerPh}.
 Similarly, $\phi_{I}$ and the pullback $\sigma \chi_{II}$ are the same up to a phase,
 \begin{equation}
    \phi_I(x) = -e^{i \beta(\tau(x))} \sigma \chi_{II}(x), \quad x \in R_I
 \end{equation}
  There exists  a negative sign in front of the phase term because of $\sigma^2 = -1$.
 One has the same relations in the region $R_{II}$,
 $$
  \chi_{II}(x) = e^{i \beta(x)}\sigma \phi_{I}(x), \quad \phi_{II}(x) = -e^{i \beta(\tau(x))} \sigma \chi_{I}(x),  \quad x \in R_{II} 
 $$
 
 \begin{lemma}
  The transition function $w: \mathbb{T} \rightarrow U(2)$ of the Hilbert bundle $\pi: \mathcal{H} \rightarrow \mathbb{T}$ is nonzero only
   at the fixed points $\mathbb{T}^\tau = \{ 0, \pi \}$.  Furthermore, its value at the fixed points are determined by the real continuous function 
   $\beta: \mathbb{T} \rightarrow \mathbb{R}$, more precisely,
    \begin{equation}
  w(0) = \begin{pmatrix}
                                0 & -e^{i \beta(0)} \\
                                e^{i \beta(0)} & 0
                              \end{pmatrix}, \quad
  w(\pi) = \begin{pmatrix}
                                0 & -e^{i \beta(\pi)} \\
                                e^{i \beta(\pi)} & 0
                              \end{pmatrix}
 \end{equation}
 \end{lemma}
  
 \begin{proof} 
 For local sections $s, t \in \Gamma(\mathbb{T}, \mathcal{H})$,  the transition function  is defined by 
 $$
 w : \mathbb{T}\rightarrow U(2); \quad w(x) = \langle s(x), \sigma t(x)\rangle  
 $$
 First let us consider $x \in R_I = (0, \pi)$, so a local basis is given by $(\phi_I, \chi_I)$.
 In the region $R_I$, the transition function is always zero.  
 Since first $\phi_I$ and $\Theta \phi_I$ are orthogonal, i.e., $w_{11}= 0$, and then
 $\phi_I$ and $\sigma \chi_{I}$ have no overlaps in their domains, so $w_{12} = 0$. 
 If we interchange the roles of $\phi_I$ and $\chi_I$, then similarly $w_{22} = 0$ and $w_{21} = 0$. 
 For the same reason, the transition function is also zero for $x \in R_{II}$.
 
 However, in a small neighborhood $O_0$ around  the fixed point $0$, the transition function is nonzero.
 Since $\phi_I$ and $\Theta \phi_I$ are orthogonal, $w_{11} = 0$, similarly $w_{22} = 0$. 
 Let us look at the off-diagonal terms, take the limit $x \rightarrow 0$,   
 $$
  \begin{array}{ll}
   \lim_{x \rightarrow 0} w_{12}(x)   & =  \lim_{x \rightarrow 0} \langle \phi_I(x), \sigma \chi_I (x) \rangle \\
                                       & =  \lim_{x \rightarrow 0} \langle \phi_I(x),   -e^{-i \beta(\tau(x))} \phi_{II}(x) \rangle \\
                                       & =  -e^{i \beta(0)} 
  \end{array}
 $$
 By a similar computation, $ \lim_{x \rightarrow 0} w_{21}(x) = e^{i \beta(0)} $. In sum,
 $$
  \lim_{x \rightarrow 0} w(x) = \begin{pmatrix}
                                0 & -e^{i \beta(0)} \\
                                e^{i \beta(0)} & 0
                              \end{pmatrix}
 $$

 Hence the transition function is skew-symmetric at the fixed point $0$,  
 \begin{equation*}
  w(0) = \begin{pmatrix}
                                0 & -e^{i \beta(0)} \\
                                e^{i \beta(0)} & 0
                              \end{pmatrix}
 \end{equation*}
 Similarly, if we use a local basis $(\phi_{II}, \chi_{II})$ in $R_{II}$ and take the limit $x \rightarrow \pi$, we have
  \begin{equation*}
  w(\pi) = \begin{pmatrix}
                                0 & -e^{i \beta(\pi)} \\
                                e^{i \beta(\pi)} & 0
                              \end{pmatrix}
 \end{equation*}
 
 \end{proof}
 
 \begin{lemma} \label{cor1}
   The Hilbert bundle $\pi: \mathcal{H} \rightarrow \mathbb{T}$ is a locally trivial twisted vector bundle, the twists only happen at 
   the fixed points $\mathbb{T}^\tau = \{ 0, \pi\}$ such that 
     \begin{equation}
  \begin{pmatrix}
    \phi \\
    \chi
  \end{pmatrix} \mapsto 
                          \begin{pmatrix}
   
                                -e^{i \beta(0)} \chi \\
                                e^{i \beta(0)} \phi
                              \end{pmatrix}, \quad
 \begin{pmatrix}
    \phi \\
    \chi
  \end{pmatrix} \mapsto  \begin{pmatrix}
                                -e^{i \beta(\pi)} \chi \\
                                e^{i \beta(\pi)} \phi
                              \end{pmatrix}
 \end{equation}
 where $(\phi, \chi)$ are local sections around   $\mathbb{T}^\tau = \{ 0, \pi \}$.
 \end{lemma}

 \begin{proof}
    This is implied from the above lemma.
 \end{proof}

 Let us look into the real phase function $\beta:  \mathbb{T} \rightarrow \mathbb{R}$ closely.
 Recall the  constraint \eqref{betaCond} on the gauge transformation $\beta$ is 
 $\beta(2\pi ) - \beta(0) = 2k \pi$  for some $k \in \mathbb{Z}$, so that 
 the Berry phase is gauge invariant. Taking time reversal symmetry into account,
 we get a refined constraint on $\beta$.

 \begin{lemma}
  Under the time reversal symmetry, the continuous real phase function $\beta: \mathbb{T} \rightarrow \mathbb{R}$ satisfies the 
  refined condition, 
  $$
  \beta(\pi) - \beta(0) = k\pi, \quad k \in \mathbb{Z}
  $$
 \end{lemma}
 \begin{proof}
   Time reversal symmetry defines the time reversal transformation $\tau$ both on $\mathbb{T}$ and $\mathbb{R}$, and the compatibility
   requires $\beta \circ \tau = \tau \circ \beta$, i.e.,
    $$
      e^{i\beta(\tau(x))} = e^{-i\beta(x)}, \quad x \in \mathbb{T}
    $$
    where $\tau: \mathbb{R} \rightarrow \mathbb{R};\,\, t \mapsto -t$ is written explicitly. 
    Now look at the fixed points,
    $$
      e^{i\beta( 0)} = e^{-i\beta(0)}, \quad  e^{i\beta(\pi)} = e^{-i\beta(\pi)}
    $$
    combine these two equations, we have
    $$
     e^{2i \beta(\pi) } =  e^{2i\beta( 0)} \quad  \Longleftrightarrow \quad 2 \beta(\pi) - 2\beta( 0) = 2 k \pi
    $$
 \end{proof}

 For example, the linear function $\beta(x) = kx$ satisfies this condition.

\begin{examp}\label{ex1}
   If the phase function is the identity map, i.e., $ \beta(x) = x $, 
 then
 $$
 w(0) = \begin{pmatrix}
                                0 & -1 \\
                                1 & 0
                              \end{pmatrix}, \quad 
 w(\pi) = \begin{pmatrix}
                                0 & 1 \\
                                -1 & 0
                              \end{pmatrix}                             
 $$

 \begin{figure}\label{fig1}
   \centering
  \includegraphics[scale = 0.4]{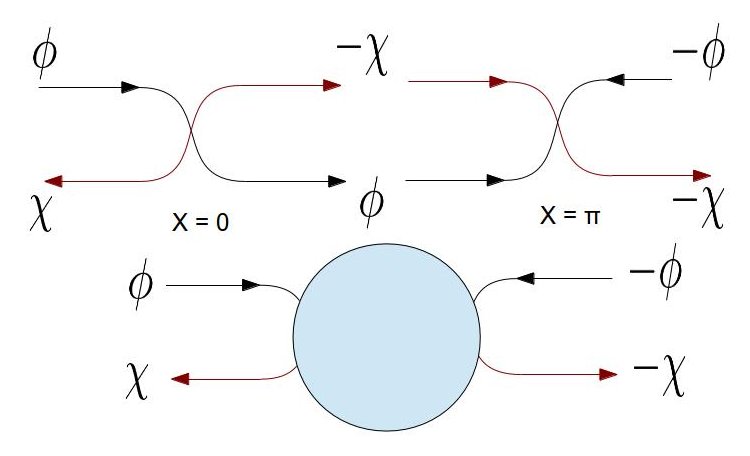}
  \caption{ Schematic illustration of Example \ref{ex1}, after an adiabatic process, i.e. evolving through the fixed points $\{ 0, \pi \}$, 
  the chiral states both change their orientations according to the transition function $w$. }
 \end{figure}
 
 By lemma \ref{cor1}, a local section $(\phi, \chi)$ is transformed at the fixed points by,
 $$ \text{at x = 0: } \,\,
 \begin{pmatrix}
    \phi \\
    \chi
  \end{pmatrix} \mapsto 
                          \begin{pmatrix}
   
                                - \chi \\
                                 \phi
                              \end{pmatrix}, \quad
                              \text{at x = $\pi$: } \,\,
 \begin{pmatrix}
    \phi \\
    -\chi
  \end{pmatrix} \mapsto  \begin{pmatrix}
                                -\chi \\
                                 -\phi
                              \end{pmatrix}
 $$
 which is illustrated in Figure \ref{fig1}.
 Notice that at the fixed point $\pi$, we changed the order of the local sections from $( -\chi, \phi)$ to $(\phi, -\chi)$
 since we used a local basis $(\phi_{II}, \chi_{II})$ to compute the transition function at $\pi$, i.e., $w(\pi)$.
 
 In this case, a global section $(\psi, \Theta \psi)$ of the Hilbert bundle $\pi: \mathcal{H}\rightarrow \mathbb{T}$ 
 is transformed to $(- \Theta \psi, -\psi)$ along the circle.
 So both chiral states, i.e., $\psi$ and $\Theta \psi$, change their orientations after traveling along the circle once.
 
 Hence, there exists an effective conical intersection made up of global sections $(\psi, \Theta\psi)$ and $(- \Theta\psi, -\psi)$.
 As a remark,  $\psi$ and $\Theta\psi$ intrinsically have opposite chirality since they are mirror image of each other under the time reversal symmetry. 
 The swap $(\psi, \Theta\psi) \mapsto (- \Theta\psi, -\psi)$ is viewed as a twist between the chiral states $\psi$ and $\Theta \psi$ in a 2d fiber, 
 which is more complicated than the 1d twist $t \mapsto -t$ in the M{\"o}bius strip.

    It is easy to compute the Kane--Mele invariant using the definition \eqref{KMinv},
     $$
     \nu = pf(w(0)) \times pf(w(\pi)) = (-1) \times 1 = -1
     $$
      As a consequence, the phase function in this example gives rise to a non-trivial topological insulator.
\end{examp}

 \begin{examp}\label{ex2}
    If the phase function is defined by $ \beta(x) = 2x$,
 then 
 $$
 w(0) = \begin{pmatrix}
                                0 & -1 \\
                                1 & 0
                              \end{pmatrix}, \quad 
 w(\pi) = \begin{pmatrix}
                                0 & -1 \\
                                1 & 0
                              \end{pmatrix}                             
 $$
 
 \begin{figure}
   \centering
  \includegraphics[scale = 0.4]{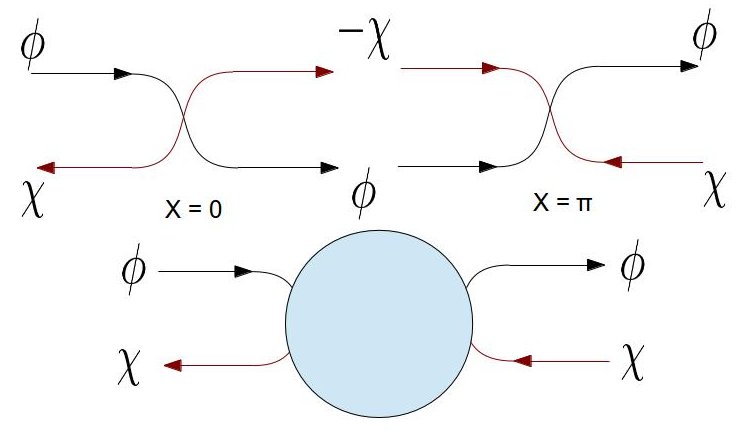}
  \caption{  Schematic illustration of Example \ref{ex2}, after an adiabatic process,   
  the chiral states retain their orientations. }
 \end{figure}
 
 Now a local section $(\phi, \chi)$ is transformed at the fixed points by,
 $$ \text{at x = 0: } \,\,
 \begin{pmatrix}
    \phi \\
    \chi
  \end{pmatrix} \mapsto 
                          \begin{pmatrix}
   
                                - \chi \\
                                 \phi
                              \end{pmatrix}, \quad
                              \text{at x = $\pi$: } \,\,
 \begin{pmatrix}
    \phi \\
    -\chi
  \end{pmatrix} \mapsto  \begin{pmatrix}
                                 \chi \\
                                  \phi
                              \end{pmatrix}
 $$
 In this case, a global section $(\psi, \Theta \psi)$ of the Hilbert bundle $\pi: \mathcal{H}\rightarrow \mathbb{T}$ 
 is transformed to $(\psi,  \Theta \psi)$ along the circle.
 So there is no conical intersection between the chiral states. 
 More precisely, at $x = 0$, one chiral state is twisted, then at $x = \pi$, the same chiral state is untwisted,
 so the effective result is no twist. 
  The Kane--Mele invariant for this phase function is 
     $$
     \nu = pf(w(0)) \times pf(w(\pi)) = (-1) \times (-1) = 1
     $$
 \end{examp}

 Since the transition function $w$ of the Hilbert bundle is skew-symmetric at the fixed points, 
 it only changes the orientation of one chiral state at one fixed point.
 In other words, since the rank of the Hilbert bundle is 2, we need two fixed points to check 
 whether the orientations of a pair of chiral states are switched or not in an adiabatic process.
 
 Because of the time reversal symmetry, it is enough to consider the effective Brillouin torus instead of the whole Brillouin torus $\mathbb{T}^d$.
 Accordingly, we define the half winding number of the gauge transformation on  the effective Brillouin torus $[0, \pi]$ in 1d.
 \begin{defn}
    The half winding number of the phase function $\beta(x): \mathbb{T} \rightarrow \mathbb{R}$, denoted by $n(\beta)$, is defined by,
    $$
     n(\beta) \,: = \frac{\beta(\pi) - \beta(0)}{\pi} = k \in \mathbb{Z}
    $$
 \end{defn}

  As a generalization of the above examples, we have the following lemma based on the refined constraint on $\beta$.
  \begin{lemma}
    The half winding number of $\beta$ determines the topological $\mathbb{Z}_2$ invariant.
     \end{lemma}
 \begin{proof}
 As mentioned in the appendix \ref{KMinv}, it is better to interpret the Kane--Mele invariant as the change of signs of Pfaffians,
 so the 1d Kane--Mele invariant is easily computed as
 $$
  \nu = \frac{ pf(w(\pi))} { pf(w(0)) } =  e^{i \beta(\pi)- i \beta(0)}  = e^{ik\pi} = (-1)^{n(\beta)}
 $$
  More precisely, for odd $k = 2n+1$, the Kane--Mele invariant $\nu(2n+1) = -1$,
    and for even $k = 2n$, the Kane--Mele invariant $\nu(2n) = +1$.   
 \end{proof}

There is an ambiguity in choosing a pair of global sections of the Hilbert bundle  as $(\psi, \Theta \psi)$ or $(\Theta \psi, \psi)$,  
so the induced Pfaffian line bundle could have $\psi$ or $\Theta \psi$ as its global section. 
In order to remove this ambiguity, one has to reduce the structure group of the Pfaffian line bundle from $U(1)$ to $\mathbb{Z}_2$.
As a result, the Kane--Mele invariant is independent of the choice of a global section, 
so  the effective physical  theory is always well-defined.

 Recall that the structure group of the Pfaffian line bundle $Pf \rightarrow \mathbb{T}$ defined in \eqref{Pfdef} is derived from that of the Hilbert bundle $\mathcal{H} \rightarrow \mathbb{T}$.
Since there is only one chiral state present in $Pf \rightarrow \mathbb{T}$ and the transition function $w$ of the Hilbert bundle is skew-symmetric, taking the Pfaffian function of 
$w$ gives us the transition function of $Pf \rightarrow \mathbb{T}$, denoted by $h = pf(w): \mathbb{T} \rightarrow U(1)$.
There are basically two different ways to reduce the structure group of  $Pf \rightarrow \mathbb{T}$ from $U(1)$ to $\mathbb{Z}_2$.

One way is to fix the initial condition of the phase function $\beta: \mathbb{T} \rightarrow \mathbb{R}$ by setting $\beta(0) = 0$ so that 
$\beta(\pi) = k\pi$. In this case, the off-diagonal matrices
$$
w(0) = \begin{pmatrix}
        0 & -1 \\
        1 & 0
       \end{pmatrix}, \quad 
w(\pi) = \begin{pmatrix}
          0 & (-1)^{k+1} \\
          (-1)^k & 0
         \end{pmatrix}
$$
give a representation of $\mathbb{Z}_2$, that is, the structure group of the Hilbert bundle is reduced from $U(2)$ to $\mathbb{Z}_2$.
At the same time, the structure group of the Pfaffian line bundle is reduced from $U(1)$ to $\mathbb{Z}_2$ since 
$$
h(0) = pf(w(0)) = -1, \quad h(\pi) = pf(w(\pi)) = (-1)^{k+1}
$$
To get the same result as above, an equivalent way is to normalize the transition functions by $e^{i\beta(0)}$, i.e., 
$$
w' = \frac{w}{e^{i\beta(0)}}, \quad h' = \frac{h}{e^{i\beta(0)}}
$$
Another commonly used way is to take the sign of Pfaffians, which is suggested by the Kane--Mele invariant,
$$
h: X^\tau \rightarrow \mathbb{Z}_2, \quad h(x) : = sgn(pf(w(x))) = \frac{pf(w(x))}{\sqrt{\det(w(x))}}  
$$
After such reduction, from now on we assume the structure group of the Pfaffian line bundle $\pi: Pf \rightarrow \mathbb{T}$ is $\mathbb{Z}_2$.
There exists a canonical principal $\mathbb{Z}_2$-bundle, i.e., the orientation bundle, denoted by $\pi: Or \rightarrow \mathbb{T}$, 
 for an open subset $U \subset \mathbb{T}$, $Or|_U \cong U \times \mathbb{Z}_2$ locally.
So the Pfaffian line bundle can be viewed as an associated bundle of the orientation bundle, i.e.,
$$
Or \times_{\mathbb{Z}_2} \mathbb{R} \cong Pf
$$
where we view the Pfaffian line bundle as a \emph{real} line bundle. As a consequence, the orientation bundle and the Pfaffian line bundle can 
be used interchangeably.

Now we can compare the Pfaffian line bundle $Pf \rightarrow \mathbb{T}$ to the M{\"o}bius strip $M\text{\"o} \rightarrow \mathbb{T}$.
The transition function of the M{\"o}bius strip is given  by \eqref{MStrans} in the appendix \ref{Mob}.
For the Pfaffian line bundle $Pf \rightarrow \mathbb{T}$, the transition function $h$ has two possibilities according
to the evenness or oddness of the half winding number of the phase function $\beta$.

\begin{examp}
  Let us continue to discuss Example \ref{ex1}, the Pfaffian line bundle has the transition function $h = pf(w)$. 
  More precisely,
  $$
  h(0) = pf(w(0)) = pf \begin{pmatrix}
                                0 & -1 \\
                                1 & 0
                              \end{pmatrix} = -1,
  $$
  $$
  h(\pi) = pf(w(\pi)) = pf \begin{pmatrix}
                                0 & 1 \\
                                -1 & 0
                              \end{pmatrix} = 1.
  $$
In other words, we have  
  \begin{equation*} 
   h:
\begin{cases}
  0 \rightarrow -1,  & \text{orientation-reversing} \\
  \pi \rightarrow +1, &  \text{orientation-preserving}
\end{cases}
\end{equation*}
which is analogous to the transition function of the M{\"o}bius strip \eqref{MStrans}.
So the effective result of $h$ is that it reverses the orientation of one chiral state that defines the Pfaffian line bundle. 
If the other chiral state was chosen to define the Pfaffian line bundle, by the transition function $w$ being skew-symmetric, 
its orientation will be preserved at $x= 0$ and reversed at $x = \pi$. In sum, the orientations of both chiral states are reversed
under $w$ along the circle in the Hilbert bundle, and this fact can be detected by looking at $h = pf(w)$ of the Pfaffian line bundle 
combined with the time reversal symmetry.

If we view $\pi: Pf \rightarrow \mathbb{T}$ as a ``real'' Pfaffian line bundle, then  its characteristic class is given by
the first Stiefel--Whitney class $w_1 \in H^1(\mathbb{T}, \mathbb{Z}_2)$. In this example, we have $w_1(Pf) = 1$ since 
the orientation of $Pf$ is reversed by the transition function, i.e.,
$h(0)h(\pi) = -1$. In other words, $w_1(Pf) = 1$ characterizes the fact that there exists a conical intersection of chiral states and 
the Kane--Mele invariant  $\nu = -1$.
\end{examp}

\begin{examp}
  Let us continue to discuss Example \ref{ex2}, the Pfaffian line bundle has the transition function $h = pf(w)$ given by,
  $$
  h(0) = pf(w(0)) = pf \begin{pmatrix}
                                0 & -1 \\
                                1 & 0
                              \end{pmatrix} = -1,
  $$
  $$
  h(\pi) = pf(w(\pi)) = pf \begin{pmatrix}
                                0 & -1 \\
                                1 & 0
                              \end{pmatrix} = -1.
  $$
In this case, both of them are orientation-reversing,
  \begin{equation*} 
   h:
\begin{cases}
  0 \rightarrow -1,  & \text{orientation-reversing} \\
  \pi \rightarrow -1, &  \text{orientation-reversing} 
\end{cases}
\end{equation*}
So the effective result of $h$ in this case is to preserve the orientation of one chiral state that defines the Pfaffian line bundle. 
If the other chiral state was chosen to define the Pfaffian line bundle, by the transition function $w$ being skew-symmetric, 
its orientation will be preserved at $x= 0$ and preserved at $x = \pi$. In sum, the orientations of both chiral states are preserved
under $w$ along the circle in the Hilbert bundle.
 In this example, we have $w_1(Pf) = 0$ since the Pfaffian bundle is trivial $Pf \cong \mathbb{T} \times \mathbb{C}$ and 
$h(0)h(\pi) = 1$. In other words, $w_1(Pf) = 0$ tells us that there is no conical intersections  and 
the Kane--Mele invariant $\nu = 1$.
\end{examp} 
  
\begin{prop}
  The first Stiefel--Whitney class of the Pfaffian line bundle $\pi: Pf \rightarrow \mathbb{T}$, i.e., 
  $w_1(Pf) \in H^1(\mathbb{S}^1, \mathbb{Z}_2) \cong \mathbb{Z}_2 $ 
  provides a method to compute the $\mathbb{Z}_2$ invariant of a 1d topological insulator. More precisely,
   its relation with the Kane--Mele invariant is given by 
   \begin{equation}
      \nu = (-1)^{w_1(Pf)}
   \end{equation}
\end{prop}
\begin{proof}
 First notice that the transition function $w$ of the Hilbert bundle gives rise to a representation of $\mathbb{Z}_2$ at the fixed points, 
 so the product $h(0)h(\pi) = pf(w(0))pf(w(\pi))$ is independent of the choice of a global section of the Pfaffian line bundle.
 Now we view the Pfaffian line bundle $\pi: Pf \rightarrow \mathbb{T}$ as a \emph{real} line bundle with the structure group $\mathbb{Z}_2$,
 so it makes sense to take the first Stiefel--Whitney class. The first Stiefel--Whitney class is the obstruction to orientability, 
 more precisely, it is computed by the product of  $h(x) =pf(w(x))$ for $x \in \mathbb{T}^\tau = \{0, \pi \}$,
 $$
    (-1)^{w_1(Pf)} = h(0)h(\pi) = pf(w(0))pf(w(\pi)) = \nu
 $$
 i.e., the Kane--Mele invariant.
 So the topological $\mathbb{Z}_2$ invariant is interpreted as the first Stiefel--Whitney class of the Pfaffian line bundle.
 In other words, the existence of a conical intersection between two orthogonal chiral states in the Hilbert bundle is translated into the orientability of the 
 Pfaffian line bundle defined by one chiral state with the help of the time reversal symmetry.
\end{proof}

\subsection{Higher dimensions}

In condensed matter physics, one considers topological insulators in spatial dimensions $1$, $2$ and $3$.
In 2d and 3d, the topological $\mathbb{Z}_2$ invariant is again the obstruction to orientability of the Pfaffian line bundle,
which is a direct generalization of the 1d case.

\begin{prop}
   The  $\mathbb{Z}_2$ invariant of a 2d (resp. 3d) topological insulator 
   can be computed by the 2nd (resp. 3rd) Stiefel--Whitney class of the Pfaffian line bundle $\pi: Pf \rightarrow X$ 
   for the momentum space $X = \mathbb{S}^n$ or $\mathbb{T}^n$ ($n = 2, 3$).
 \end{prop}
 
\begin{proof}
   Let us first consider the sphere $X = \mathbb{S}^n$ ($n = 2, 3$), which has 2 fixed points, i.e., 
   $(\mathbb{S}^n)^\tau = \{ N, S \}$ called the north and south poles.    
One defines the Pfaffian line bundle $\pi: Pf \rightarrow \mathbb{S}^n$ as before,
which has one chiral state as a global section. $\pi: Pf \rightarrow \mathbb{S}^n$ is a locally trivial complex line bundle
with possible twists at the fixed points. 
The structure group of the Pfaffian line bundle is $\mathbb{Z}_2$ and its transition function is $h: (\mathbb{S}^n)^\tau = \{N, S \} \rightarrow \mathbb{Z}_2$.

When a  chiral state undergoes an adiabatic process, it sweeps out the entire $n$-sphere once, so instead of a closed 1-cycle, 
one has to use a closed $n$-cycle when computing the Berry phase. 
Hence the orientability of the Pfaffian line bundle $\pi: Pf \rightarrow \mathbb{S}^n$ is determined by the top Stiefel--Whitney class 
$w_n(Pf) \in H^n(\mathbb{S}^n, \mathbb{Z}_2) \cong \mathbb{Z}_2$, and the Kane--Mele invariant is computed by 
$$
\nu = h(N)h(S) = (-1)^{w_n(Pf)}
$$ 

When the momentum space $X = \mathbb{T}^n$ ($n =2, 3$), one has the collapse map $q: \mathbb{T}^n \rightarrow \mathbb{S}^n$.  
It induces a map in cohomology $q^*: H^n(\mathbb{S}^n, \mathbb{Z}_2) \rightarrow H^n(\mathbb{T}^n, \mathbb{Z}_2)$, so that 
the top cohomology of $\mathbb{T}^n$ is the same as that of $\mathbb{S}^n$. So the topological $\mathbb{Z}_2$ invariant over $\mathbb{T}^n$
is determined by that over $\mathbb{S}^n$, and both of them can be computed by the top Stiefel--Whitney class of the Pfaffian line bundle.
We will give another proof for the torus based on cobordsim theory in the next section.
\end{proof}
 
Even though higher order Stiefel--Whitney classes are employed in understanding the topological $\mathbb{Z}_2$ invariant,
there is no other geometric structures involved other than orientability.
This is determined by the $\mathbb{Z}_2$ structure group of the Pfaffian line bundle and the $\mathbb{Z}_2$-CW complex structure of the momentum space.
Physically, an adiabatic process involves the whole momentum space, i.e., an $n$-cycle, instead of a $1$-cycle.

\begin{examp}

In general, for a complex line bundle $\mathcal{L} \rightarrow X$, say $X$ is 2d,
 there exists a relation between the first Chern class of  $\mathcal{L}$ 
and the second Stiefel--Whitney class of $\mathcal{L}$ viewed as a real vector bundle,
$$
 H^2(X, \mathbb{Z}) \rightarrow H^2(X, \mathbb{Z}_2), \quad c_1(\mathcal{L}) \equiv w_2(\mathcal{L}) \quad \text{(mod 2)} \quad \forall\,\, \mathcal{L} \rightarrow X
$$

In our case, the Pfaffian line bundle $\pi: Pf \rightarrow X$ ($X = \mathbb{S}^2$ or $\mathbb{T}^2$) is  a complex line bundle, 
so it is characterized by the first Chern class $c_1(Pf) \in H^2(X, \mathbb{Z})$.
In fact,  $c_1(Pf)$ is a 2-torsion
because of the time reversal symmetry, i.e., $c_1(Pf) \in H^2(X, \mathbb{Z}_2)$, for details see \cite{KLW1501}.
So the second Stiefel--Whitney class is equal to the first Chern class,
$$
w_2(Pf) = c_1(Pf) \in  H^2(X, \mathbb{Z}_2) = \mathbb{Z}_2
$$

\end{examp}

\section{Cobordism theory}\label{Cobthy}

In the previous section, we mainly dealt with the Pfaffian line bundle over the sphere $\mathbb{S}^n$, 
we will take care of the torus $\mathbb{T}^n$ in this section. 
Based on the $\mathbb{Z}_2$-CW complex structure of the involutive space $(\mathbb{T}^n, \tau)$, the 
torus can be viewed as a cobordism connecting sub-tori (possibly with extra structures), 
so the topological $\mathbb{Z}_2$ invariant has a new interpretation in this context.
Physically, an adiabatic process is divided into sub-processes, and each sub-process will contribute to the collective effect.

\subsection{Unoriented cobordism}\label{unorCob}
A cobordism is a natural equivalence relation on the class of compact $n$-manifolds defined by a $(n+ 1)$-manifold with boundary.
By definition, two compact $n$-manifolds $M$ and $N$ are cobordant when there exists a compact $(n+1)$-manifold $W$ whose boundary is the disjoint union 
of $M$ and $N$, i.e., $\partial W = M \sqcup N$. The class of manifolds cobordant to
$M$ is called the cobordism class of $M$, denoted by $[M]$.   

Let $ {\mathfrak {N}}_{n}$ denote the set of cobordism classes of closed unoriented n-manifolds, 
it is an abelian group with the addition defined by the disjoint union or connected sum, $[M] + [N] = [M \sqcup N] = [M \# N]$. Moreover, 
${\mathfrak {N}}_{*} = \sum_{n \geq 0}{\mathfrak {N}}_{n}$ is a graded ring with the multiplication defined by the Cartesian product,
$[M] \times [N] = [M \times N]$. It is well-known that 
$$
{\mathfrak {N}}_{*} = \mathbb{F}_2[x_i ~|~ i \geq 1, i \neq 2^j -1 ] = \mathbb{F}_2[x_2, x_4, x_5, \cdots]
$$
where $\mathbb{F}_2$ is the field with 2 elements. 

In particular, $\mathfrak{N}_2 = \mathbb{Z}_2$ is generated by the real projective plane 
$x_2 =[ \mathbb{RP}^2]$. Two unoriented closed surfaces $S_1$ and $S_2$ are cobordant if and only if their Stiefel--Whitney numbers are the same,
i.e., $\langle w_2(TS_1), [S_1]\rangle = \langle w_2(TS_2), [S_2] \rangle $, where $TS$ is the tangent bundle of a closed surface $S$.
By the relation between the top Stiefel--Whitney class and the Euler class, $w_2(TS)$ can be replaced by the mod 2 Euler class $e(TS)$.
In fact, the mod 2 Euler characteristic $\chi: \mathfrak{N}_2 \rightarrow \mathbb{Z}_2$ is a group isomorphism.
A Klein bottle $\mathbb{K}$ can be constructed as the connected sum of two real projective planes, i.e.,
$\mathbb{K}= \mathbb{RP}^2 \# \mathbb{RP}^2 $.  The Euler characteristic of the real projective plane is $\chi(\mathbb{RP}^2) = 1$ 
and that of the Klein bottle is $\chi(\mathbb{K}) = 0$. In addition, the 2-torus $\mathbb{T}^2$ 
is an orientable closed surface with Euler characteristic $\chi(\mathbb{T}^2) = 0$.

\subsection{2-torus}
Let us first consider the 2-torus $\mathbb{T}^2$, which has 4 fixed points of the time reversal transformation $\tau$.
We pair the nearest neighbor fixed points into two pairs,
and each pair produces a circle by gluing an 1d $\mathbb{Z}_2$-cell. Let us refer to the resulting circles as the north  and south circles, 
denoted by $\mathbb{T}_N$ and $\mathbb{T}_S$, since by convention we think of a cobordism vertically.
These two circles are always cobordant since $\partial (\mathbb{T} \times I) = \mathbb{T} \sqcup \mathbb{T}$, 
and in 2d the effective Brillouin zone is given by $\mathbb{T} \times I$ where $I$ is the unit interval. 
In our case, a circle such as $\mathbb{T}_N$ or $\mathbb{T}_S$   is endowed with an extra structure, i.e., the Pfaffian line bundle, 
so the cobordism class is more interesting.

First of all, the Pfaffian line bundle over a circle is viewed as a \emph{real} line bundle with the structure group $\mathbb{Z}_2$. 
$$
\xymatrix{
Pf \ar[d]^\pi \ar[r]^{f^*} & E\mathbb{Z}_2 \ar[d]^\pi\\
\mathbb{T} \ar[r]^f        & B\mathbb{Z}_2}
$$
Based on the classifying space $B\mathbb{Z}_2 = BO(1) = \mathbb{RP}^\infty$, such Pfaffian line bundle is characterized by the classifying map
$f:  \mathbb{T} \rightarrow B\mathbb{Z}_2 $, i.e., 
$$
[f]\in [ \mathbb{S}^1,  BO(1)] = \pi_1(BO(1)) = \pi_0(O(1)) = \pi_0(\mathbb{S}^0)= \mathbb{Z}_2
$$
This is another way to determine whether the Pfaffian line bundle is orientation-preserving or orientation-reversing in homotopy theory,
which is equivalent to the first Stiefel--Whitney class in cohomology with $\mathbb{Z}_2$ coefficients.

If we restrict the Pfaffian line bundle $\pi: Pf \rightarrow \mathbb{T}^2$ onto a circle $\mathbb{T}$, then we obtain the first 
Stiefel--Whitney class $w_1(Pf|_\mathbb{T}) \in H^1(\mathbb{S}^1, \mathbb{Z}_2)$ characterizing the orientability of the restricted Pfaffian bundle
$Pf|_\mathbb{T} \rightarrow \mathbb{T}$. In our case, we have $w_1(Pf|_{\mathbb{T}_N})$ and $w_1(Pf|_{\mathbb{T}_S})$ for the north and south circles.
In order to use unoriented cobordism \ref{unorCob}, we compactify the Pfaffian line bundle into a circle bundle over $\mathbb{T}$, i.e.,
a closed surface. 
When $w_1(Pf|_\mathbb{T}) = 0$, the restricted Pfaffian bundle is isomorphic to a cylinder, so its compactified bundle is identified with the 2-torus $\mathbb{T}^2$. 
When $w_1(Pf|_\mathbb{T}) = 1$, the restricted Pfaffian bundle is isomorphic to a M{\"o}bius strip, 
and its compactified bundle  is the real projective plane $\mathbb{RP}^2$. 

From the knowledge of the second unoriented cobordism group $\mathfrak{N}_2 = \mathbb{Z}_2$, $[\mathbb{RP}^2]$ and $[\mathbb{T}^2]$ does not belong to
the same cobordism class since they have different Euler characteristics. As a result, 
the  restricted Pfaffian bundles $Pf|_{\mathbb{T}_N}$ and $Pf|_{\mathbb{T}_S}$ are not always cobordant.
In other words, given two circles $\mathbb{T}_N$ and $\mathbb{T}_S$, each is endowed with a principal $\mathbb{Z}_2$-bundle,
 is it possible to find a $2$-manifold $W$ bounded by $\mathbb{T}_N \sqcup \mathbb{T}_S$
together with a principal $\mathbb{Z}_2$-bundle $E \rightarrow W$ so that its restriction to the boundary recovers $Pf|_{\mathbb{T}_N}$ and $Pf|_{\mathbb{T}_S}$?
The answer to this question depends on the orientability of $Pf|_{\mathbb{T}_N}$ and $Pf|_{\mathbb{T}_S}$. 
Indeed, $Pf|_{\mathbb{T}_N}$ and $Pf|_{\mathbb{T}_S}$ are cobordant if and only if $w_1(Pf|_{\mathbb{T}_N}) = w_1(Pf|_{\mathbb{T}_S})$.

Based on the skeleton decomposition of $(\mathbb{T}^2, \tau)$,
an adiabatic process over $\mathbb{T}^2$ can be divided into two sub-processes over $\mathbb{T}_N$ and $\mathbb{T}_S$. 
The local orientabilities of the restricted Pfaffian bundles over $\mathbb{T}_N$ and $\mathbb{T}_S$ together determine
the global orientability of the Pfaffian line bundle $\pi: Pf \rightarrow \mathbb{T}^2$, in terms of Stiefel--Whitney classes 
we have the following relation,
\begin{equation}\label{w2relw1}
   w_2(Pf) \equiv w_1(Pf|_{\mathbb{T}_N}) + w_1(Pf|_{\mathbb{T}_S}) \quad \text{mod 2}
\end{equation}
since orientability is a collective effect. In the left hand side of the above formula, 
it looks like one might use  the restricted Pfaffian bundle $\pi: Pf|_{\mathbb{T} \times I} \rightarrow \mathbb{T} \times I$
 over the effective Brillouin torus $\mathbb{T} \times I$ instead of $\pi: Pf \rightarrow \mathbb{T}^2$. 
However, the key observation is that the local orientability of the Pfaffian line bundle is completely determined by its 
transition function at the fixed points, see the computation in the next paragraph,
and then the restricted Pfaffian bundle can be extended trivially to the whole Brillouin torus $\mathbb{T}^2$.

Recall that $h = pf(w)$ is the transition function of the Pfaffian line bundle $\pi: Pf \rightarrow X$.
 We  interpret the Kane--Mele invariant as the product of change of  $h$  over distinct pairs of nearest neighbor fixed points \ref{KMinv},
$$
 \nu  = \prod_{(x, y) \in X^\tau \times X^\tau} h(x)h(y)  
$$
where a pair of fixed points $(x_1,y_1) $ is different form $(x_2, y_2)$ if $x_1, y_1, x_2, y_2$ are all different.
Hence the relation \eqref{w2relw1} can be expressed in terms of Kane--Mele invariants,
$$
 \nu({\mathbb{T}^2}) = \nu({\mathbb{T}_N}) \nu({\mathbb{T}_S})  
$$
since 
$$
\begin{array}{ll}
   (-1)^{w_2(Pf)}& =  {h(\Gamma)h(A)h(B)h(C)}=  [{h(\Gamma)h(A)}] [{h(B)h(C)}] \\
                  & = (-1)^{w_1(Pf|_{\mathbb{T}_N})} (-1)^{w_1(Pf|_{\mathbb{T}_S})}

\end{array}
$$
where $\Gamma = (0, 0), A = (\pi, 0), B = (0, \pi)$ and $C =(\pi, \pi)$ are the fixed points in the 2-torus. 

\begin{prop}
   The topological $\mathbb{Z}_2$ invariant of the 2-torus $\mathbb{T}^2$ can be understood by unoriented cobordism.
   More precisely, the restricted Pfaffian bundles being cobordant is the obstruction to the orientability of the Pfaffian line bundle $\pi:Pf \rightarrow \mathbb{T}^2$, i.e.
   the triviality of the topological $\mathbb{Z}_2$ invariant $\nu(\mathbb{T}^2)$.
 \end{prop}

\begin{proof}
  The Pfaffian line bundle $\pi: Pf \rightarrow \mathbb{T}^2$ is orientable if and only if $w_2(Pf) = 0 \in H^2(\mathbb{T}^2, \mathbb{Z}_2)$, which is equivalent to
  $w_1(Pf|_{\mathbb{T}_N}) = w_1(Pf|_{\mathbb{T}_S})$, 
  that is, $Pf|_{\mathbb{T}_N}$ and $Pf|_{\mathbb{T}_S}$ are cobordant. 
  From the relations, 
  $$
  \nu(\mathbb{T}_N) = (-1)^{w_1(Pf|_{\mathbb{T}_N})}, \quad \nu(\mathbb{T}_S) = (-1)^{w_1(Pf|_{\mathbb{T}_S})}
  $$
  the Kane--Mele invariant $\nu(\mathbb{T}^2) = \nu(\mathbb{T}_N)\nu(\mathbb{T}_S) = 1$ if and only if the Pfaffian line bundle $\pi: Pf \rightarrow \mathbb{T}^2 $ has cobordant restrictions,
   i.e., $[Pf|_{\mathbb{T}_N} ] = [Pf|_{\mathbb{T}_S} ]$ as cobordism classes.
   In other words, the orientability of the Pfaffian line bundle $\pi: Pf \rightarrow \mathbb{T}^2$ is equivalent to 
   the triviality of the topological $\mathbb{Z}_2$ invariant. As a remark, from the relation \eqref{w2relw1} the orientability of the Pfaffian line bundle
   $\pi: Pf \rightarrow \mathbb{T}^2$ is independent of the choice of $\mathbb{T}_N$ and $\mathbb{T}_S$, i.e., different pairing of nearest neighbor fixed points.
\end{proof}

In \cite{MB07}, the authors studied the homotopy type of mappings from the effective Brillouin zone in 2d, i.e., a cylinder, to the space of time reversal invariant Hamiltonians 
as a subspace of Fredholm operators, denoted by $\mathcal{Q} \subset \mathbf{Fred}$.
They first capped the cylinder by disks to make it into a 2-sphere,  then used the homotopy group $\pi_2(\mathcal{Q}) = [S^2, \mathcal{Q}]$
to understand the topological $\mathbb{Z}_2$ invariant.
For a general momentum space $X \neq S^n$, homotopy theory is replaced by K-theory, since, for example, the space of Fredholm operators
gives a classifying space of complex K-theory: $K(X) = [X, \mathbf{Fred}]$. For the subspace $\mathcal{Q}$ over the 2-torus, 
we have to consider 
the Quaternionic K-theory, i.e.,  $KQ(\mathbb{T}^2)$ \cite{FM13, KLW1501}.
By the isomorphism $KQ(\mathbb{T}^2) \cong KR^{-4}(\mathbb{T}^2) \cong KO^{-2}(pt) = \mathbb{Z}_2$, the topological $\mathbb{Z}_2$ invariant is 
living in the real K-group $KO^{-2}$. There exists a natural transformation between cobordism theory and K-theory, for example, 
$\mu: \Omega_{SU}^* \rightarrow KO^*$. 
It would be interesting to explicitly work out the isomorphism map between the cobordism theory and the K-theory of the topological $\mathbb{Z}_2$ invariant.

\subsection{3-torus}
Similar to the 2-torus, the topological $\mathbb{Z}_2$ invariant of the 3-torus can be also reduced to 
the behavior of restricted Pfaffian bundles over circles.  The 3-torus has 8 fixed points of the time reversal symmetry, we
pair nearest neighbor fixed points together as before to get four pairs, each pair produces a circle by gluing an 1d $\mathbb{Z}_2$-cell, denoted by $\mathbb{T}_i$, $i = 1, 2,3,4$.
Each circle is endowed with a restricted Pfaffian bundle $Pf|_{\mathbb{T}_i} \rightarrow \mathbb{T}_i$ by restricting $\pi: Pf \rightarrow \mathbb{T}^3$, and each restricted Pfaffian bundle 
is isomorphic to either a cylinder or a M{\"o}bius strip. In addition, the orientability of each restricted Pfaffian bundle is characterized
by its first Stiefel--Whitney class $w_1(Pf|_{\mathbb{T}_i})$. 

Next we divide those four pairs of fixed points into two sets, each set has two pairs of distinct fixed points. There are three different
ways to divide 4 pairs into 2 sets, which corresponds to three different embeddings of a 2-torus into the 3-torus. 
For example, we choose $\{ \mathbb{T}_1, \mathbb{T}_2\}$ as the first set, denoted by $N$, and then the second set $S = \{ \mathbb{T}_3, \mathbb{T}_4 \}$.
We will see later that the topological $\mathbb{Z}_2$ invariant is independent of the choice of such sets. 
As mentioned above, the involutive space $(\mathbb{T}^2, \tau)$ has a $\mathbb{Z}_2$-CW complex structure, 
so we can construct a 2-torus by gluing
 $\mathbb{Z}_2$-cells onto $N = \{ \mathbb{T}_1, \mathbb{T}_2\}$, denoted it by $\mathbb{T}^2_N$, 
similarly we get the other 2-torus $\mathbb{T}^2_S$ from the set $S$. 
$\mathbb{T}^2_N$ and $\mathbb{T}^2_S$ are viewed as two sub-tori embedded into $\mathbb{T}^3$.

From the previous subsection, 
the orientability of the restricted Pfaffian bundle  $Pf|_{\mathbb{T}^2_N} \rightarrow {\mathbb{T}^2_N} $ is determined by
the relation
$$
w_2(Pf|_{\mathbb{T}^2_N}) \equiv w_1(Pf|_{\mathbb{T}_1}) +  w_1(Pf|_{\mathbb{T}_2}) \quad \text{mod 2}
$$
and $Pf|_{\mathbb{T}^2_N}$ is orientable if and only if $Pf|_{\mathbb{T}_1}$ and $Pf|_{\mathbb{T}_2}$ are cobordant.
Similarly, 
$$
w_2(Pf|_{\mathbb{T}^2_S}) \equiv w_1(Pf|_{\mathbb{T}_3}) +  w_1(Pf|_{\mathbb{T}_4}) \quad \text{mod 2}
$$

Putting together, we have a similar relation for the Pfaffian line bundle over the 3-torus $\pi: Pf \rightarrow \mathbb{T}^3$,
\begin{equation*}
   w_3(Pf) \equiv  w_2(Pf|_{\mathbb{T}^2_N}) + w_2(Pf|_{\mathbb{T}^2_S})  \equiv \sum_{i =1 }^4 w_1(Pf|_{\mathbb{T}_i}) \quad \text{mod 2}
\end{equation*}
From the right hand side of the above congruence relation, $w_3$ is independent of the choice of  the sets of circles $N$ and $S$.
In other words, the orientability of $\pi: Pf \rightarrow \mathbb{T}^3$ is entirely determined by the restricted Pfaffian bundle
around the fixed points $(\mathbb{T}^3)^\tau$. 

\begin{prop}
   The relation between the weak and strong topological insulators is given by 
   $$
     \nu(\mathbb{T}^3) = \nu(\mathbb{T}^2_N) \nu(\mathbb{T}^2_S)
   $$
   or equivalently, in terms of Stiefel--Whitney classes
   \begin{equation*}
   w_3(Pf) \equiv  w_2(Pf|_{\mathbb{T}^2_N}) + w_2(Pf|_{\mathbb{T}^2_S})    \quad \text{mod 2}
\end{equation*}
where $\mathbb{T}^2_N$ and $\mathbb{T}^2_S$ are two sub-2-tori containing distinct 4 fixed points, i.e., $(\mathbb{T}^2_N)^\tau \cap (\mathbb{T}^2_S)^\tau = \emptyset$.
\end{prop}
 
\begin{proof}
   A 3d topological insulator is said to be weak if restricted to a sub-2-torus the Kane--Mele invariant is non-trivial, 
   i.e., $\nu(\mathbb{T}^2) = -1$, but in fact it is an ordinary insulator, i.e., $\nu(\mathbb{T}^3) = 1 $.
   In other words,  the restricted Pfaffian bundle $Pf|_{\mathbb{T}^2}$ is not orientable, but   $\pi: Pf \rightarrow \mathbb{T}^3$ is orientable.
   This situation can only happen when $\nu(\mathbb{T}^2_N) = \nu(\mathbb{T}^2_S) = -1$ or equivalently when $ w_2(Pf|_{\mathbb{T}^2_N}) = w_2(Pf|_{\mathbb{T}^2_S}) =1  $
   for a certain pair of 2-tori $\mathbb{T}^2_N$ and $\mathbb{T}^2_S$.
   
   In addition, the above weak-strong relation can be expressed in terms of cobordism.
   Indeed,  $[Pf|_{\mathbb{T}^2_N}] = [Pf|_{\mathbb{T}^2_S}] \neq 0$ are non-trivial cobordism classes 
   if and only if $\pi: Pf \rightarrow \mathbb{T}^3$ is orientable but its restrictions to $Pf|_{\mathbb{T}^2_N}$ and $Pf|_{\mathbb{T}^2_S}$ are not orientable, 
   that is, the total topological $\mathbb{Z}_2$ invariant is trivial $\nu(\mathbb{T}^3) = 1$, but $\nu(\mathbb{T}^2_N) = \nu(\mathbb{T}^2_S) = -1$.
\end{proof}
 
\section{Extended TQFT}\label{extTQFT}
The topological $\mathbb{Z}_2$ invariant over the 3-torus gives rise to a toy model of a  fully extended $3$-dimensional topological quantum field theory (TQFT).

Fix an $n$-manifold $M$, let $Bord_{M}$ be the bordism category whose morphisms are $k$-dimensional $(k \leq n)$ submanifolds of $M$ and
whose objects are connected components of the boundaries of such submanifolds. In other words, if $M$ can be decomposed into $0$-, $1$-, $\cdots$, $n$-manifolds with corners,
then the category $Bord_{M}$ has objects compact $0$-manifolds, $1$-morphisms compact $1$-manifolds
with boundary, $2$-morphisms compact $2$-manifolds with corners, and so on.
By definition, an extended $n$-dimensional TQFT on $M$ is a symmetric monoidal functor from $Bord_{M}$ to the category of vector spaces,
$$
F: (Bord_{M}, \sqcup) \rightarrow (\mathbf{Vect}, \otimes)
$$
In general, one needs an $(\infty, n)$-category $\mathcal{C}$ to define an extended TQFT as a homomorphism $F: Bord_{M} \rightarrow \mathcal{C}$.
However, in our case we only use the category of vector spaces $\mathbf{Vect}$, so we call it a toy model of an extended TQFT.

\begin{examp}
Let us first define an extended TQFT  on $\mathbb{T}^3$, then we will state the theorem following this example. 
The $\mathbb{Z}_2$-CW complex structure of $(\mathbb{T}^3, \tau)$ can be derived from that of the circle $(\mathbb{T}, \tau)$. Roughly, 
if we only decompose one circle as a $\mathbb{Z}_2$-CW complex, 
$$
\mathbb{T}^3 = \mathbb{T}^2 \times (\mathbb{T}^\tau \sqcup D^1 \times \mathbb{Z}_2) = \mathbb{T}^2 \times \mathbb{S}^0 \sqcup \mathbb{T}^2 \times D^1 \times \mathbb{Z}_2 
    = (\mathbb{T}^2 \times I) \curvearrowleft \mathbb{Z}_2
$$
then $\mathbb{T}^3$ (modulo the $\mathbb{Z}_2$ action) can be viewed as a cobordism (i.e., $\mathbb{T}^2 \times I$) with two sub-2-tori as the boundary,
since $\mathbb{T}^\tau = \mathbb{S}^0$ is fixed by the $\mathbb{Z}_2$ action and $D^1 = I \setminus \mathbb{S}^0$ has a free $\mathbb{Z}_2$ action.
Physically, this reduction is equivalent to taking the effective Brillouin zone of $\mathbb{T}^3$, i.e., $\mathbb{T}^2 \times I$,  since the time reversal $\mathbb{Z}_2$ symmetry
can recover the omitted mirror image easily. If we think of the cobordism vertically, then we call one $\mathbb{T}^2$ the north torus, denoted by $\mathbb{T}^2_N$, 
and the other the south torus, denoted by $\mathbb{T}^2_S$.

With the Pfaffian line bundle $\pi: Pf \rightarrow \mathbb{T}^3$, we can define a functor,  called the partition function,
$$
Z: Bord_{\mathbb{T}^3} \rightarrow \mathbf{Vect}_{\mathbb{F}_2}
$$
such that 
$Z(\mathbb{T}^3) = w_3(Pf) \in H^3(\mathbb{T}^3, \mathbb{Z}_2) = \mathbb{Z}_2$ and 
$Z(\mathbb{T}^2_i) = w_2(Pf|_{\mathbb{T}^2_i}) \in H^2({\mathbb{T}^2_i}, \mathbb{Z}_2)= \mathbb{Z}_2$ for $i = N, S$.
By definition, $Z$ is a homomorphism that maps an $n$-submanifold to its top cohomology group with $\mathbb{Z}_2$ coefficients, i.e., $Z(-)=H^n(-,\mathbb{Z}_2)$. 

We similarly define the partition function $Z$ on $\mathbb{T}^2$ with the help of the Pfaffian line bundle $\pi: Pf \rightarrow \mathbb{T}^2$.
First we view $\mathbb{T}^2$ (modulo the $\mathbb{Z}_2$ action) as a cobordism, that is,  
we reduce $\mathbb{T}^2$ to $\mathbb{T} \times I$ by the time reversal symmetry, 
 $$
\mathbb{T}^2 = \mathbb{T} \times (\mathbb{T}^\tau \sqcup D^1 \times \mathbb{Z}_2) = \mathbb{T} \times \mathbb{S}^0 \sqcup \mathbb{T} \times D^1 \times \mathbb{Z}_2 
     = (\mathbb{T} \times I)  \curvearrowleft \mathbb{Z}_2
$$
So  $\mathbb{T} \times I$
connects the north circle and the south circle, 
denoted by $\mathbb{T}_N$ and $\mathbb{T}_S$,  i.e., $\partial (\mathbb{T} \times I) = \mathbb{T}_N \sqcup \mathbb{T}_S$.
Next, we define the partition function $Z :  Bord_{\mathbb{T}^2} \rightarrow \mathbf{Vect}_{\mathbb{F}_2}$ as above such that 
$Z(\mathbb{T}^2) = w_2(Pf) \in H^2(\mathbb{T}^2, \mathbb{Z}_2) = \mathbb{Z}_2$ and 
$Z(\mathbb{T}_i) = w_1(Pf|_{\mathbb{T}_i}) \in H^1({\mathbb{T}_i}, \mathbb{Z}_2)= \mathbb{Z}_2$ for $i = N, S$.

Finally, we define the partition function $Z$ on the circle $(\mathbb{T}, \tau)$. First we view $\mathbb{T}$ (modulo the $\mathbb{Z}_2$ action) as a cobordism 
$$
\mathbb{T} = \mathbb{S}^0 \sqcup D^1 \times \mathbb{Z}_2  = I \curvearrowleft \mathbb{Z}_2
$$
such that $\partial I = \mathbb{S}^0 = \{ N, S \}$ consisting of the north pole and the south pole.
With the Pfaffian line bundle $\pi: Pf \rightarrow \mathbb{T}$, we define $Z(\mathbb{T}) = w_1(Pf) \in H^1(\mathbb{T}, \mathbb{Z}_2) = \mathbb{Z}_2 $ and
$Z(pt) = h(pt) \in \mathbb{Z}_2$  for $pt \in \{ N, S \}$ where $h$ is the transition function of the Pfaffian line bundle $\pi: Pf \rightarrow \mathbb{T}$,
 $$
 h \in H^0(\mathbb{T}^\tau, \mathbb{Z}_2) = H^0(\mathbb{S}^0, \mathbb{Z}_2) = H^0( \mathbb{Z}_2, \mathbb{Z}_2) \cong  \mathbb{Z}_2
 $$  
\end{examp}

\begin{thm}
 The topological $\mathbb{Z}_2$ invariant defines the partition function of an extended TQFT over the  momentum space $(X, \tau)$, 
 $$
 \nu: (Bord_{X}, \sqcup) \rightarrow (\mathbf{Vect}_{\mathbb{F}_2}, \otimes) 
 $$
\end{thm}
\begin{proof}
    Based on the band structure of a topological insulator,
   one has the Pfaffian line bundle over the momentum space, $\pi: Pf \rightarrow X$. 
   Since the momentum space is closed, we assume that the set of fixed points $X^\tau$ is finite and  has even points.
   Because of the $\mathbb{Z}_2$-CW complex structure of $(X, \tau)$, we view $X$ as a bordism category $Bord_X$. Indeed, 
   the reverse process of the  $\mathbb{Z}_2$-CW construction by gluing $\mathbb{Z}_2$-cells onto the fixed points $X^\tau$ gives the cobordism decomposition of $Bord_X$.
   
   Now we define the partition function as 
   $$
     \nu : Bord_X \rightarrow \mathbf{Vect}_{\mathbb{F}_2}; \quad M \mapsto \nu(M) = (-1)^{w_i(Pf|_M)}
   $$
   Since $w_i(Pf|_M) \in H^i(M, \mathbb{Z}_2)$ for an $i$-submanifold $M$, the range of the Kane--Mele invariant is the exponentiated cohomology groups, 
   the effective result is to convert the addition operation between Stiefel--Whitney classes into an multiplication operation.
   However, with the $\mathbb{Z}_2$ coefficients, the exponentiated additive group $\mathbb{Z}_2$ is isomorphic to the multiplicative group $\mathbb{Z}_2$,
   so we identify the range of $\nu$ as the cohomology groups as before, but with multiplication instead of addition,
   i.e., elements in the monoidal category $(\mathbf{Vect}_{\mathbb{F}_2}, \otimes) $ (vector spaces over the field $\mathbb{F}_2$ with the monoidal operation $\otimes$).
   
   From the above example, we have seen that the partition function defined by the Stiefel--Whitney classes of the Pfaffian line bundle restricted to proper submanifolds 
   gives rise to an extended TQFT. The partition function $Z(-) = w_i(Pf|_{-})$ is a symmetric monoidal functor since
   $$
   Z(M \sqcup N) = w_i(Pf|_{M \sqcup N} ) = w_i(Pf|_M) + w_i(Pf|_N)
   $$ 
   if $M$ and $N$ are $i$-manifolds. By the relation
   $\nu(M) = (-1)^{Z(M)}$, $\nu$ is also a symmetric monoidal functor.
   
   Notice that this extended TQFT is completely determined by the $0$-manifolds, i.e., the fixed points $X^\tau$,
   since the Kane--Mele invariant is determined by $X^\tau$,
   $$
   \nu(X) = \prod_{x \in X^\tau} h(x)
   $$ where $h$ is the transition function of the Pfaffian line bundle $\pi: Pf \rightarrow X$.
\end{proof}

In fact, the above extended TQFT is a 3-2-1-0 Chern--Simons theory with the gauge group $\mathbb{Z}_2$.
To spell it out, one needs all characterizations of the topological $\mathbb{Z}_2$ invariant \cite{KLW15}. 
This project will be carried out in a future paper.

\appendix
\section{Appendices} \label{Appdx}
 \subsection{M{\"o}bius strip} \label{Mob}

Let us briefly review some basic facts about the M{\"o}bius strip, which will be viewed as a real line bundle over the circle with non-trivial Stiefel--Whitney class.
The M{\"o}bius strip is an open surface with boundary, capping the  M{\"o}bius strip by a 2-disc, one could obtain the real projective plane $\mathbb{RP}^2$ or the Klein bottle. 

The open M{\"o}bius strip is a non-orientable surface with Euler characteristic $\chi = 0$, which can be constructed by gluing the two ends of a paper strip after a half-twist. 
As a vector bundle, the M{\"o}bius strip is a real line bundle over the circle, i.e., $\pi: M\text{\"o} \rightarrow \mathbb{S}^1$, 
with a typical fiber $F = M\text{\"o}_x \cong \mathbb{R}$.
If one focuses on orientability, then the M{\"o}bius strip can be easily distinguished from the cylinder, i.e., a trivial real line bundle over the circle.

 Let $U$ and $V$ be two open subsets that cover $\mathbb{S}^1$, i.e., $U \cup V = \mathbb{S}^1$.
For example, if $\mathbb{S}^1$ is parametrized by the angle $\theta$, one can take $U = (0 - \epsilon, \pi + \epsilon)$ and $V = (-\pi - \epsilon, 0 + \epsilon)$ for a fixed small number $\epsilon > 0$.
 The M{\"o}bius strip $M\text{\"o}$ is  locally trivial, 
that is, $M\text{\"o}|_U \cong U \times \mathbb{R}$ and  $M\text{\"o}|_V \cong V \times \mathbb{R}$. However,  $M\text{\"o}$ is globally non-trivial, 
which is characterized by the transition function $g : U \cap V \rightarrow GL(1, \mathbb{R})$.
Indeed, $U \cap V = (0-\epsilon, 0+\epsilon) \cup (\pi - \epsilon, \pi+ \epsilon)$
if we identify $\pi$ with $-\pi$, let $O_0 = (0-\epsilon, 0+\epsilon) $ and $O_\pi = (\pi - \epsilon, \pi+ \epsilon)$,
so that $U \cap V = O_0 \cup O_\pi$.  
By the construction of the M{\"o}bius strip, there exists one half-twist before gluing the two ends of a paper strip. 
If a fiber $F \cong \mathbb{R}$ has the coordinate $t$, 
then the transition function, denoted by $g$, is the identity map on one overlap, say $O_0$,
$$
g : O_0 \rightarrow GL(1, \mathbb{R}); \quad g(\theta) = 1: \mathbb{R} \rightarrow \mathbb{R},\,\, t \mapsto t
$$
Whereas $g$ is the minus identity map on the other overlap,
$$
g : O_\pi \rightarrow GL(1, \mathbb{R}); \quad g(\theta) = -1: \mathbb{R} \rightarrow \mathbb{R},\,\, t \mapsto -t
$$
which gives the half-twist mentioned above.

Since we only care about orientability of a real line bundle, we consider the structure group $\mathbb{Z}_2$ instead of $GL(1, \mathbb{R}) = \mathbb{R}^\times$.
So the transition function of $M\text{\"o}$ is reduced to $g: \{0, \pi \} \rightarrow \mathbb{Z}_2$ since $\lim O_0 = 0 $
and $\lim O_\pi = \pi$ as $\epsilon \rightarrow 0$, 
\begin{equation}\label{MStrans}
   g:
\begin{cases}
  0 \rightarrow +1,  & \text{orientation-preserving} \\
  \pi \rightarrow -1, & \text{orientation-reversing}
\end{cases}
\end{equation}
Therefore, the M{\"o}bius strip $M\text{\"o} \rightarrow \mathbb{S}^1$ is a non-orientable bundle 
since the transition function $g$ reverses the local orientations of $M\text{\"o}$ between different coordinate patches.

\subsection{Stiefel--Whitney class}
As an example of a $\mathbb{Z}_2$-characteristic class, the Stiefel--Whitney class provides a  topological invariant of a real vector bundle. 
Given a real vector bundle $\pi: E \rightarrow X$, its total Stiefel--Whitney class, denoted by $w(E)$, is an element of the 
cohomology ring $H^*(X, \mathbb{Z}_2)$ with coefficients in $\mathbb{Z}_2$. If the cohomology ring $H^*(X, \mathbb{Z}_2)$ is decomposed into
cohomology groups with different degrees, then one has the $i$-th Stiefel--Whitney class of $E$, denoted by $w_i(E) \in H^i(X, \mathbb{Z}_2)$, 
$$
 H^*(X, \mathbb{Z}_2) = \oplus_{i = 0}^n H^i(X, \mathbb{Z}_2); \quad w(E) = 1 + w_1(E) + w_2(E) + \cdots + w_n(E)
$$
where $n$ is the rank of $E$.
The Stiefel--Whitney classes can be constructed by Steenrod square operations and the Thom isomorphism, 
which can be realized on real Grassmannians $Gr_n = Gr(n, \mathbb{R}^\infty)$ \cite{MS74}.

The Stiefel--Whitney classes play important roles in obstruction theory, that is, the vanishing of a Stiefel--Whitney class is the obstruction to some geometric structure.
For example, the first Stiefel--Whitney class establishes an isomorphism between real line bundles and the first cohomology group,
$$
w_1: Vect_1(X) = [X, Gr_1] \rightarrow H^1(X, \mathbb{Z}_2)
$$
In addition, the first Stiefel--Whitney class detects the orientability of a real bundle $E \rightarrow X$, 
that is,  $w_1(E) = 0$ if and only if $E$ is orientable.
In other words, the first Stiefel--Whitney class $w_1$ is the obstruction to orientability.
As an example, the M{\"o}bius strip $\pi: M\text{\"o} \rightarrow \mathbb{S}^1$ is non-orientable and
$w_1(M\text{\"o}) = 1 \in H^1(\mathbb{S}^1, \mathbb{Z}_2)$. 
In general, a manifold $M$ is called orientable if its tangent bundle is orientable, i.e., $w_1(TM) = 0$. 

Higher order Stiefel--Whitney classes are obstructions to other interesting geometric structures.
A real bundle $E \rightarrow X$ admits a spin structure if and only if $w_1(E) = w_2(E) = 0$.
So for an orientable real bundle, the second Stiefel--Whitney class $w_2$ is the obstruction to spin structure.
A manifold $M$ is a spin manifold if its tangent bundle admits a spin structure.
For an orientable real bundle $E \rightarrow X$, the second Stiefel--Whitney class $w_2(E) \in H^2(X, \mathbb{Z}_2)$ is in the image of the mod 2 map
$H^2(X, \mathbb{Z}) \rightarrow H^2(X, \mathbb{Z}_2)$. This fact is equivalent to  the vanishing of the third integral Stiefel--Whitney class $W_3(E) \in H^3(X, \mathbb{Z})$. 
In addition, $W_3(E) = 0$ if and only if $E$ admits a spin$^c$ structure. 
In other words,  the third integral Stiefel--Whitney class $W_3$
is the obstruction to spin$^c$ structure on an orientable real bundle. 
For an oriented real vector bundle $E \rightarrow X $ with rank $n$, 
one defines the Euler class $e(E) \in H^n(X, \mathbb{Z})$ based on the orientation class,
the top Stiefel--Whitney class is the Euler class ignoring orientation, i.e., $w_n(E) \equiv e(E) \,\, \text{(mod 2)}$,
$$
 H^n(X, \mathbb{Z}) \rightarrow H^n(X, \mathbb{Z}_2); \quad e(E) \mapsto w_n(E)
$$

\subsection{Berry phase} \label{BerPh}
Given a time reversal invariant Hamiltonian $H$, we want to study its band structure, i.e., the topological band theory in the language of vector bundles.
 In fact, the above Hamiltonian $H$ is parametrized by the momentum space, i.e., a family of Hamiltonians $H(x)$ for  $x \in X$. 
 So one should expect a Berry phase 
 when an electronic state undergoes a cyclic adiabatic process. 
 Berry phase is the prototypical example of a geometric phase, 
 which can be computed by a contour integral of the Berry connection over a closed cycle.
 Let us briefly review the Berry connection and Berry curvature, and derive a useful constraint of a geometric phase.
 
 The Berry connection is defined by a covariant derivative $D$ on a Hilbert bundle $\pi: \mathcal{H} \rightarrow X$,
 for simplicity, let us assume $D$ equals the exterior derivative $d$.
 For a section $s \in \Gamma(X, \mathcal{H})$, the local form of the Berry connection is given by
 $  A(x) :=  \langle s(x), ds(x) \rangle  $
 and the Berry curvature is defined by $F = dA$ with the local form $F(x) : =  d\langle s(x), ds(x) \rangle$.
 Under a gauge transformation, the Berry connection $A$ is gauge dependent and the Berry curvature $F$ is gauge invariant.

 For a closed cycle $C$, the Berry phase is defined as 
 $$
 \gamma : = i \oint_C A
 $$
 If the closed cycle $C$ forms the boundary of a surface $S$, then by the Stoke's theorem,
 $$
 \gamma =i \int_S F
 $$
 Furthermore, if the surface $S$ happens to be a momentum space,
 then the Berry phase gives us the first Chern number, 
 $$
 c_1 = \frac{1}{2\pi}\, \gamma = \frac{i}{2\pi}  \int_{S} F
 $$
 
  When applying a gauge transformation to a section $s \in \Gamma(X, \mathcal{H})$,
 $$
  s(x) \mapsto e^{-i\beta(x)} s(x)
 $$
 where $\beta: X \rightarrow \mathbb{R}$ is a continuous real-valued function,
 the Berry connection is not gauge invariant, so it is transformed to 
 $$
 A(x) \mapsto A(x) - i d\beta(x)
 $$
 As a result, the Berry phase acquires an extra term derived from the phase function $\beta$,
 $$
 \gamma \mapsto \gamma +  \oint_C  d\beta
 $$
 Therefore,   $\beta$ has to satisfy the constraint,
 $$
  \oint_C  d\beta \equiv 0 \quad \text{mod} \,\, 2\pi 
 $$
 since $exp\{ i \gamma \}$ is supposed to be gauge invariant. If we reparametrize the coordinate of $C$ by the time $t$,
 the above condition is then expressed as 
 \begin{equation}\label{betaCond}
     \beta(T) - \beta(0) = 2k\pi, \quad  k \in \mathbb{Z}
 \end{equation}
 provided that the adiabatic process
 is completed in the time interval $ [0, T]$.
 
 \subsection{Kane--Mele invariant}\label{KMinv}
 The topological $\mathbb{Z}_2$ invariant has many equivalent characterizations from different perspectives, 
 here we use the Kane--Mele invariant as its definition.
 The Kane--Mele invariant was originally derived from the quantum spin Hall effect on graphene \cite{FK06}, 
 and later generalized to three dimensional topological insulators. 
 Let us briefly review the Kane--Mele invariant and explain its physical meaning in this appendix. 
 
 One considers the Hilbert bundle $\pi: \mathcal{H} \rightarrow X$ in \S2.1 characterized by the transition function $w$, 
 and the set of fixed points of the time reversal symmetry is $X^\tau$.
 At any fixed point $x\in X^\tau$,  the transition function $w$ from eq.\eqref{TransFunc} becomes a skew-symmetric matrix, i.e., $ w^T(x) =- w ({x})$, where $T$ is the transpose of a matrix. So it makes sense to
take the Pfaffian function of $w$ at any fixed point $x \in X^\tau$, denoted by $pf[w(x)]$.
 The Kane--Mele invariant $\nu$ is defined as the product of the signs of Pfaffians over the fixed points, 
\begin{equation}  
   \nu := \prod_{x_i \in X^\tau} sgn(pf[w(x_i)])    
\end{equation} 
Recall the relation between the Pfaffian and determinant functions is given by 
$pf^2(A) = \det(A)$  for a skew-symmetric matrix $A$, so that $pf(A) = \pm \sqrt{\det(A)}$.
By comparing with the square root of the determinant of $w$ at $x \in X^\tau$, i.e., $\sqrt{\det[w(x)]}$,
the sign of Pfaffians can be written as the ratio,
\begin{equation*} 
   \nu  = \prod_{x_i \in X^\tau} \frac{pf [w(x_i)]}{\sqrt{\det [w(x_i)]}}   
\end{equation*} 
When the Kane--Mele invariant $\nu = -1$, the material is a non-trivial topological insulator, and when 
$\nu = 1$, the material is a classical insulator.
 
 Geometrically, the Pfaffian function is promoted to a section of the   Pfaffian line bundle $Pf \rightarrow X$, 
 and taking the sign of Pfaffians reduces the structure group of the Pfaffian line bundle to  $\mathbb{Z}_2$. 
 This construction gives rise to the orientation bundle $Or \rightarrow X$ associated with the Pfaffian line bundle, 
 which is a canonical principal $\mathbb{Z}_2$-bundle.  The Pfaffian bundle is orientable if and only if its orientation
 bundle is trivializable, that is, there exists a trivialization such that $Or \cong X \times \mathbb{Z}_2$.
 
 The Kane--Mele invariant $\nu$ collects the information about all possible twists running through the fixed points.
 A negative sign from $sgn(pf[w(x_i)]) $ for some fixed point $x_i \in  X^\tau$ is a twist that reverses the local orientation 
 of the Pfaffian line bundle. Furthermore, another twist will reverses the local orientation around another fixed point,
 so the effective result of two negative signs,
 i.e., two twists, is to preserve the orientation. On the other hand, a positive sign from $sgn(pf[w(x_i)]) $
 for some fixed point $x_i \in  X^\tau$ does not change the local orientation, i.e., orientation-preserving. 
 The product of all signs of Pfaffians, i.e., the Kane--Mele invariant $\nu$, determines whether the global orientation is
 preserved ($\nu = +1$) or reversed ($\nu = -1$).
 
 If we pair nearest neighbor fixed points together, then the Kane--Mele invariant is interpreted as the product of change of signs of Pfaffians
 over pairs of fixed points,
 $$
 \nu = \prod_{(x_i, x_j) \in X^\tau \times X^\tau} sgn(pf[w(x_i)])   sgn(pf[w(x_j)])   
 $$
 Indeed, the change of signs of Pfaffians is more important than the signs of Pfaffians themselves, 
 since once the sign is changed between a pair of fixed points, the local orientation is reversed under an adiabatic (sub)process.

 
\nocite{*}
\bibliographystyle{plain}
\bibliography{SW}

\end{document}